\documentclass{article}

\usepackage{arxiv}

\usepackage[utf8]{inputenc}
\usepackage[T1]{fontenc}
\usepackage{hyperref}
\usepackage{url}
\usepackage{booktabs}
\usepackage{amsfonts}
\usepackage{nicefrac}
\usepackage{microtype}
\usepackage{graphicx}
\usepackage{amsmath}
\usepackage{amsthm}
\usepackage{enumitem}

\graphicspath{ {./images/} }

\newtheorem{proposition}{Proposition}
\newtheorem{corollary}{Corollary}

\title{The Economics of Builder Saturation in Digital Markets}

\author{
 Armin Catovic \\
  Director Data \& AI, Funnel\\
  Stockholm, Sweden \\
  \texttt{armin.catovic@funnel.io}, \texttt{armin@divergent.icu} \\
}

\begin{document}
\maketitle

\begin{abstract}
Recent advances in generative AI have dramatically reduced the cost of 
digital production, fueling narratives that widespread participation in 
software creation will yield a proliferation of viable companies. This 
paper challenges that assumption. We introduce the \emph{Builder Saturation 
Effect}, formalizing a model in which production scales elastically but 
human attention remains finite. In markets with near-zero marginal costs 
and free entry, increases in the number of producers dilute average 
attention and returns per producer, even as total output expands. Extending 
the framework to incorporate quality heterogeneity and reinforcement 
dynamics, we show that equilibrium outcomes exhibit declining average payoffs 
and increasing concentration, consistent with power-law-like distributions. 
These results suggest that democratised production is more likely to intensify 
competition and produce winner-take-most outcomes than to generate broadly 
distributed entrepreneurial success.

\medskip\noindent\textbf{Contribution type.} This paper is primarily a work of 
synthesis and applied formalisation. The individual theoretical ingredients---attention 
scarcity, free-entry dilution, superstar effects, preferential attachment---are 
well established in their respective literatures. The contribution is to combine 
them into a unified framework and direct the resulting predictions at a specific 
contemporary claim about AI-enabled entrepreneurship.
\end{abstract}


\section{Introduction}
\label{sec:introduction}

Recent advances in artificial intelligence have dramatically reduced the cost and complexity of creating digital products. In February 2025, AI researcher Andrej Karpathy coined the term ``vibe coding'' to describe a mode of software development in which users describe desired functionality in natural language and accept AI-generated code with minimal review~\cite{karpathy2025vibe}. Within a year the practice moved from novelty to mainstream: the 2025 Stack Overflow Developer Survey reports that 84\% of developers use or plan to use AI coding tools~\cite{stackoverflow2025}, GitHub's own data show that 46\% of all new code on its platform is AI-generated~\cite{github2025octoverse}, and in Y~Combinator's Winter 2025 batch, 25\% of admitted startups had codebases that were 95\% or more AI-generated~\cite{mehta2025yc}.

These developments have fuelled a widely circulated narrative: that the barriers to building are collapsing and, as a consequence, the number of successful companies will increase dramatically. OpenAI CEO Sam Altman has stated that ``you'll have billion-dollar companies run by two or three people with AI''~\cite{altman2025conversations}, and a broader discourse anticipates a future in which nearly every individual can participate as a builder in the digital economy.

This paper argues that such claims conflate an expansion of productive capacity with a proportional expansion of realised value. The critical constraint they overlook is human attention. As Herbert~A.\ Simon observed, ``a wealth of information creates a poverty of attention''~\cite{simon1971}. In digital environments where the marginal cost of reproducing information goods approaches zero~\cite{shapiro1999}, attention---not production---becomes the binding scarce resource. If aggregate attention does not grow commensurately with the number of producers, the result is not broadly distributed success but intensified competition for a finite resource.

Existing evidence from digital markets is consistent with this view. The Apple App Store hosts approximately 1.9~million apps, yet close to a quarter have fewer than 100 downloads~\cite{businessofapps2026apple}. Revenue concentration is extreme: the top 1\% of monetising publishers on the U.S.\ App Store capture approximately 94\% of all revenue, while the top 1\% of all publishers account for 70\% of total downloads~\cite{sensortower2016}. The average smartphone user engages with roughly 10~apps per day and 30~per month~\cite{buildfire2026}---a figure that has remained stable for years despite continuous growth in supply. On GitHub, despite 36~million new developers joining the platform in 2025, maintainers report being overwhelmed by AI-generated contributions that the platform's own analysis likens to a ``denial-of-service attack on human attention''~\cite{infoq2026github}. These patterns---elastic supply, inelastic attention, and concentrated outcomes---are precisely those predicted by the model developed in this paper.

The theoretical ingredients are individually well established. Models of monopolistic competition~\cite{spence1976,dixit1977} show that free entry can generate excessive product proliferation, particularly when products are close substitutes. Rosen's theory of superstars~\cite{rosen1981} demonstrates that small quality differences produce disproportionate reward differences in scalable markets. Stochastic models of preferential attachment~\cite{simon1955,barabasi1999} generate the heavy-tailed distributions observed empirically. Network-effects models~\cite{katz1985} explain user lock-in and switching costs. The contribution of this paper is to synthesise these mechanisms into a single attention-constrained entry framework directed at a specific contemporary claim: that dramatically lower build costs imply broadly distributed entrepreneurial success. While each ingredient is individually known, no existing paper derives the joint prediction that equilibrium per-builder returns are invariant to demand-side expansion (Proposition~\ref{prop:zero_profit}) and that this invariance persists even under endogenous attention growth (Proposition~\ref{prop:endogenous})---results with direct implications for the current discourse on AI-enabled entrepreneurship.

We introduce the \emph{Builder Saturation Effect} and show that in digital markets with near-zero marginal production costs, increasing the number of builders leads to (i)~a systematic dilution of average attention and returns per builder, and (ii)~a transition toward heavy-tailed outcome distributions under realistic assumptions of heterogeneity and reinforcement. These results suggest that democratised production is more likely to intensify competition and produce winner-take-most outcomes than to generate broadly distributed success. We use ``effect'' rather than ``law'' deliberately: the prediction is a robust tendency of the model under stated assumptions, not a claimed universal constant.

The remainder of the paper proceeds as follows. Section~\ref{sec:related} surveys related work. Section~\ref{sec:model} introduces the formal model of attention allocation under free entry; Section~\ref{sec:heterogeneity} extends the model to incorporate heterogeneity and reinforcement dynamics. Section~\ref{sec:discussion} discusses implications. Section~\ref{sec:conclusion} concludes. Appendix~\ref{sec:appendix} provides additional propositions and proofs.


\section{Related Works}
\label{sec:related}

The argument developed in this paper sits at the intersection of attention economics, the economics of information goods, industrial organization, and network-based models of cumulative advantage. Across these literatures, a common theme emerges: when the supply of artifacts grows faster than the capacity of users to evaluate and adopt them, competitive dynamics shift from production scarcity to attention scarcity, and realized outcomes become increasingly concentrated.

\subsection{Attention as the Scarce Resource}

The most direct antecedent to the present framework is Herbert Simon's account of informational abundance and attentional scarcity. Simon argued that in an information-rich world, the scarce factor is no longer information itself but the attention required to process it, with the implication that organizations and markets must allocate attention carefully rather than merely maximize output. This observation provides the conceptual foundation for treating aggregate user attention as the binding constraint in digital markets.

This attention-based perspective is especially relevant for contemporary digital production environments because software and media markets permit rapid expansion in the number of available artifacts without a comparable expansion in users' cognitive bandwidth. In that setting, the main competitive margin becomes discovery, evaluation, and retention rather than sheer capacity to produce. Simon's framing therefore supplies the basic scarcity principle underlying the Builder Saturation Effect.

\subsection{Information Goods and Near-Zero Marginal Reproduction Cost}

A second foundational literature comes from the economics of information goods. Shapiro and Varian emphasize that information goods are characterized by high fixed costs of initial creation and very low marginal costs of reproduction and distribution. That cost structure makes explosive entry possible once tools reduce creation costs, but it does not remove demand-side scarcity. Instead, it tends to intensify competition over users, standards, switching costs, and installed base advantages.

This distinction is central to the present paper. The claim is not that digital production cannot scale; rather, it is that production can scale much more elastically than attention. The information-goods literature helps explain why a fall in build costs can generate a large rise in entry without implying proportional gains in average producer returns.

\subsection{Product Proliferation and Monopolistic Competition}

The industrial-organization literature provides a formal basis for linking easier entry to excessive product proliferation. Spence's model of product selection under monopolistic competition shows that free-entry equilibria need not coincide with socially optimal product variety, particularly when products are close substitutes. Additional entrants may largely redistribute demand across similar offerings rather than create equivalent new surplus.

Dixit and Stiglitz extend this line of analysis by formalizing optimum product diversity under monopolistic competition with scale economies. Their framework became a canonical model for thinking about how market equilibria can generate too much or too little variety relative to the social optimum, depending on the structure of preferences and costs. For the present theory, these models matter because they show that an increase in the number of producers does not, by itself, demonstrate an increase in welfare or viability. In highly substitutable domains, more entry may simply thin demand across a larger set of offerings.

This literature therefore supports a key component of the Builder Saturation view: as barriers to entry fall, competitive markets can become crowded in ways that reduce average returns even if total output and formal variety continue to expand.

\subsection{Superstar Markets and Convex Reward Structures}

A closely related body of work concerns the concentration of rewards in scalable markets. Rosen's theory of superstars explains how small differences in talent, quality, or performance can translate into disproportionately large differences in income when production is scalable and consumers prefer the highest-quality supplier. In such settings, market expansion need not democratize rewards; it can instead magnify inequality among producers.

This insight is particularly relevant for digital goods, where one successful product can often serve a very large user base at low incremental cost. The present paper builds on that logic by arguing that once attention becomes the scarce input, the ease of entry does not flatten competition but may instead sharpen it, with a small number of products capturing outsized shares of demand. Rosen's framework thus provides the economic counterpart to the concentration result in our model.

\subsection{Skew Distributions, Preferential Attachment, and Cumulative Advantage}

The mathematical shape of these concentrated outcomes is studied in the literature on skew distributions and preferential attachment. Simon's earlier work on skew distribution functions provided one of the classic stochastic accounts of highly unequal outcome distributions, showing how simple generative processes can produce heavy tails.

Barab\'{a}si and Albert later gave a network-based account of similar phenomena, showing that growing systems in which new nodes preferentially attach to already well-connected nodes generate scale-free degree distributions. Their model provides a tractable representation of ``rich-get-richer'' dynamics, in which early advantage and reinforcement amplify inequality over time.

These models are highly relevant to digital markets because user adoption, visibility, and integration often reinforce themselves. Products that gain early traction become easier to discover, more trustworthy, more compatible with complements, and more likely to attract further users. The present paper adopts this cumulative-advantage intuition to explain why attention dilution and concentration can coexist: average returns may fall as entry rises, while realized demand becomes increasingly dominated by a minority of products.

\subsection{Winner-Take-Most Dynamics and Condensation}

Beyond standard preferential attachment, Bianconi and Barab\'{a}si show that competitive network systems can display distinct phases, including ``fit-get-rich'' and winner-takes-all behavior, and draw an analogy to Bose--Einstein condensation in physics~\cite{bianconi2001competition,bianconi2001bose}. In their framework, under sufficiently strong reinforcement and fitness heterogeneity, one node can capture a macroscopic share of links.

This result is conceptually useful for the present theory because it provides a physics-inspired language for phase transitions in market concentration. The contribution of the current paper is not to claim literal physical equivalence, but to use the condensation analogy to describe how digital markets may shift from broad experimentation to highly asymmetric allocation of attention and value once entry grows large relative to the available attention budget.

\subsection{Network Effects, Lock-In, and Installed Base Advantages}

A further strand of related work emphasizes network effects and compatibility. Katz and Shapiro show that in markets with network externalities, the value of adoption depends on the size of the installed base, and compatibility choices can strongly affect market structure~\cite{katz1985}. These mechanisms help explain why even when entry is cheap, users may cluster around a small number of products or standards.

This is directly relevant to the notion of ``inertia'' motivating the present paper. Users do not choose among new entrants in a vacuum; they face switching costs, coordination needs, learned workflows, and compatibility constraints. As a result, the outside option is not simply ``any other new product,'' but often ``stay with the incumbent ecosystem.'' Network-effects models therefore help justify the inclusion of outside-option and reinforcement terms in the formal framework developed here.

\subsection{Congestion and Contest Analogies}

Finally, the present paper also relates to congestion and contest frameworks. Congestion games formalize situations in which multiple agents compete over a shared resource whose value declines as more agents use it. Rosenthal's classic result shows that such games possess pure-strategy Nash equilibria, making them a useful analogy for entry into crowded attention markets.

Likewise, contest-success-function models provide a way to think about how effort or quality translates into probabilistic shares of a prize when agents compete for a scarce reward. While the present paper does not adopt a full rent-seeking model, contest formulations are closely related to the share-allocation rule used in our framework, in which each producer's realized demand depends on its attractiveness relative to competing alternatives.

\subsection{Contribution Relative to Existing Literature}

Existing work has separately explained attentional scarcity, excessive product variety, 
superstar concentration, network reinforcement, and winner-take-most dynamics. 
Individually, none of these results is new. The contribution of this paper is primarily 
one of \emph{synthesis and application}: we combine these mechanisms into a single 
attention-constrained entry framework and direct it at a specific contemporary claim---that 
dramatically lower build costs imply a future of broadly distributed entrepreneurial success.

Recent work has significantly advanced the formal economics of attention. 
Loewenstein and Wojtowicz~\cite{loewenstein2025} provide a comprehensive survey 
of attention in economic theory, identifying attentional externalities introduced 
by digital technology and the potential impact of AI on attention allocation as 
key open directions---precisely the questions this paper addresses in a market-entry 
context. Cattaneo et al.~\cite{cattaneo2024} develop a nonparametric Attention 
Overload Model in which each alternative's attention frequency decreases as the 
choice set grows, providing revealed-preference foundations for the intuition that 
more options dilute per-option attention. Their micro-level result complements our 
market-level analysis: we show that the same dilution logic, combined with free 
entry, pins equilibrium returns at $k/p$ independently of market size. 
Cattaneo et al.~\cite{cattaneo2020ram} introduce a Random Attention Model 
formalising how limited and stochastic attention shapes observable choice, 
further establishing the microeconomic foundations on which our 
aggregate-market framework builds.

The Builder Saturation Effect is therefore best understood not as a novel theoretical primitive, 
but as a \emph{named regularity} that emerges from the interaction of well-established 
components when applied to the current regime of near-zero marginal production costs. 
Its value lies in making explicit a structural tension that existing narratives tend to 
overlook: the divergence between elastic production and inelastic attention.


\section{Model}
\label{sec:model}

This section introduces a minimal formal framework to capture the interaction between elastic production and finite attention. The objective is not to model all aspects of digital markets, but to isolate the core mechanism underlying the Builder Saturation Effect: the divergence between scalable production and bounded consumption capacity.

\subsection{Environment}
\label{sec:environment}

Consider a population of $N$ agents. A subset $M \leq N$ acts as consumers, while a subset $B \leq N$ acts as builders (producers). For simplicity, we allow overlap between these roles but treat them analytically as distinct.

Each consumer is endowed with a fixed attention budget $a > 0$, representing the limited capacity to evaluate, adopt, or engage with products over a given period. Aggregate available attention in the system is therefore:
\begin{equation}
    A = M \cdot a
    \label{eq:aggregate_attention}
\end{equation}

This attention budget is the central scarce resource in the model.

Builders produce digital artifacts (e.g., applications, tools, services). Consistent with the economics of information goods~\cite{shapiro1999}, production is characterized by:
\begin{itemize}
    \item a fixed cost of entry $k > 0$,
    \item negligible marginal cost of reproduction $c \approx 0$.
\end{itemize}

Thus, once a product is created, it can serve additional users without significant additional cost.

\subsection{Attention Allocation}
\label{sec:attention_allocation}

Consumers allocate their attention across available products and an outside option. The outside option captures inertia, including non-adoption, incumbent usage, or status quo bias.

Let $q_i \in \mathbb{R}$ denote the quality (or attractiveness) of product $i$, and let $q_0$ denote the attractiveness of the outside option.

We assume that aggregate attention allocated to product $i$ follows a standard discrete-choice (logit) form:
\begin{equation}
    s_i = A \cdot \frac{e^{\beta q_i}}{\sum_{j=1}^{B} e^{\beta q_j} + e^{\beta q_0}}
    \label{eq:logit}
\end{equation}
where:
\begin{itemize}
    \item $s_i$ is the total attention captured by product $i$,
    \item $\beta > 0$ measures sensitivity to quality differences.
\end{itemize}

This formulation captures two key features:
\begin{enumerate}
    \item \textbf{Relative competition}: attention depends on how a product compares to alternatives.
    \item \textbf{Outside-option competition}: products must also compete against non-engagement.
\end{enumerate}

\subsection{Symmetric Benchmark}
\label{sec:symmetric}

To establish a baseline, consider a symmetric case where all products have identical quality:
\begin{equation}
    q_i = q \quad \forall\, i
\end{equation}

Then attention is evenly distributed across products and the outside option:
\begin{equation}
    s_i = \frac{A}{B + z}
    \label{eq:symmetric_attention}
\end{equation}
where:
\begin{equation}
    z = e^{\beta (q_0 - q)}
    \label{eq:outside_option}
\end{equation}
represents the effective weight of the outside option.

This yields a simple expression for average attention per builder:
\begin{equation}
    \bar{s}(B) = \frac{A}{B + z}
    \label{eq:avg_attention}
\end{equation}

From this, it immediately follows that:
\begin{equation}
    \frac{d\bar{s}(B)}{dB} < 0
\end{equation}

That is, \textbf{average attention per builder decreases monotonically as the number of builders increases}. This is the core dilution mechanism.

\subsection{Builder Payoffs and Entry}
\label{sec:payoffs}

Each builder monetizes attention at rate $p > 0$. Profit for builder $i$ is given by:
\begin{equation}
    \pi_i = p \cdot s_i - k
    \label{eq:profit}
\end{equation}

Under symmetry:
\begin{equation}
    \pi(B) = p \cdot \frac{A}{B + z} - k
    \label{eq:symmetric_profit}
\end{equation}

We assume free entry: builders enter until expected profit is driven to zero. The equilibrium number of builders $B^*$ therefore satisfies:
\begin{equation}
    p \cdot \frac{A}{B^* + z} = k
    \label{eq:zero_profit_condition}
\end{equation}

Solving:
\begin{equation}
    B^* = \frac{pA}{k} - z
    \label{eq:equilibrium_B}
\end{equation}

Since $B^*$ represents a count of producers, we impose the constraint:
\begin{equation}
    B^* = \max\!\left(\frac{pA}{k} - z,\; 0\right)
    \label{eq:equilibrium_B_constrained}
\end{equation}

The boundary case $B^* = 0$ obtains when $k \geq \frac{pA}{z}$, i.e., when 
fixed costs are sufficiently high relative to monetizable attention that no entry 
is viable. In such regimes, the outside option absorbs all available attention. 
The interior solution $B^* > 0$ requires:
\begin{equation}
    k < \frac{pA}{z}
    \label{eq:entry_viability}
\end{equation}
which is the \emph{entry viability condition}. Note that as AI-assisted tools 
drive $k \to 0$, this condition is satisfied for any positive attention pool, 
confirming that cost reduction removes supply-side barriers to entry without 
addressing demand-side constraints.

This expression yields several comparative statics:
\begin{itemize}
    \item $\frac{\partial B^*}{\partial A} > 0$: more total attention supports more builders
    \item $\frac{\partial B^*}{\partial p} > 0$: higher monetization increases entry
    \item $\frac{\partial B^*}{\partial k} < 0$: lower entry costs increase entry
\end{itemize}

In particular, a reduction in $k$---as enabled by AI-assisted production---leads to an increase in equilibrium entry $B^*$.

However, at equilibrium, profits are zero by construction:
\[
\pi(B^*) = 0
\]

Thus, \textbf{lower entry costs increase participation but do not increase average realized profit}. Instead, they intensify competition for a fixed attention pool.

\subsection{Builder Saturation}
\label{sec:saturation}

Combining the attention allocation and free-entry condition yields the central result:

\begin{quote}
As the number of builders increases relative to total available attention, average attention per builder declines, and equilibrium entry adjusts such that expected profits are driven toward zero.
\end{quote}

In the limit as $B \to \infty$:
\[
\bar{s}(B) \to 0
\]

That is, average realized attention per builder vanishes.

This establishes the first component of the Builder Saturation Law: \textbf{attention dilution under elastic entry}.

\subsection{Extension: Heterogeneity and Reinforcement}
\label{sec:heterogeneity}

The symmetric benchmark abstracts from quality differences and dynamic feedback. To capture more realistic market behaviour, we introduce two extensions: (1)~heterogeneous quality, with $q_i$ drawn i.i.d.\ from a distribution $F$ with support on $[\underline{q},\bar{q}]$; and (2)~reinforcement dynamics, in which adoption depends on both quality and existing popularity.

\subsubsection{Attention Dynamics}

Let $x_i(t) \geq 0$ denote the attention stock of product $i$ at time $t$, subject to the aggregate constraint:
\begin{equation}
\sum_{i=1}^{B} x_i(t) + x_0(t) = A \quad \forall\, t
\label{eq:attention_constraint}
\end{equation}
where $x_0(t)$ is the residual attention absorbed by the outside option.

At each discrete time step, a fraction $\delta \in (0,1]$ of total attention $A$ is reallocated. This fraction represents users who switch products, new users entering the market, or existing users reassessing their choices. Each unit of reallocatable attention is assigned to product $i$ with probability:
\begin{equation}
p_i(t) = \frac{x_i(t)^{\alpha}\, e^{\beta q_i}}{\displaystyle\sum_{j=1}^{B} x_j(t)^{\alpha}\, e^{\beta q_j} \;+\; x_0(t)^{\alpha}\, e^{\beta q_0}}
\label{eq:reallocation}
\end{equation}
where $\alpha \geq 0$ governs the strength of reinforcement (preferential attachment) and $\beta > 0$ governs sensitivity to intrinsic quality differences. The outside option enters symmetrically, preserving the role of inertia from the baseline model.

The deterministic mean-field update rule is:
\begin{equation}
x_i(t+1) = (1 - \delta)\, x_i(t) + \delta\, A\, p_i(t) \qquad \forall\, i \in \{0, 1, \ldots, B\}
\label{eq:update}
\end{equation}

The stochastic version---in which each of the $\delta A$ reallocated units is drawn independently according to $p_i(t)$---converges to the deterministic system in the large-$A$ limit by standard law-of-large-numbers arguments.

\subsubsection{Nested Special Cases}

Equation~\eqref{eq:reallocation} nests several known models:
\begin{itemize}
    \item $\alpha = 0$: static logit allocation (Section~\ref{sec:attention_allocation}), in which attention depends only on quality.
    \item $\alpha = 1$, homogeneous $q_i = q$: standard linear preferential attachment~\cite{barabasi1999}, which generates power-law degree distributions $P(x) \sim x^{-3}$ in the large-$B$ limit.
    \item $\alpha = 1$, heterogeneous $q_i$: the Bianconi--Barab\'{a}si fitness model~\cite{bianconi2001competition}, which produces power laws with fitness-dependent exponents and, under sufficient heterogeneity, condensation (winner-take-all) phases~\cite{bianconi2001bose}.
\end{itemize}

\subsubsection{Imported Analytical Results}

We state the key distributional results from the cited literature and explain their economic interpretation in the present setting. These propositions are \emph{not} novel results of this paper; they are imported from the network-science literature and applied to our attention-allocation framework.

\begin{proposition}[Power law under homogeneous reinforcement; imported from~\cite{barabasi1999}]
\label{prop:imported_powerlaw}
When $\alpha = 1$ and $q_i = q$ for all $i$, the stationary distribution of attention shares follows a power law $P(x) \propto x^{-3}$ in the limit $B \to \infty$.
\end{proposition}

\noindent\textit{Interpretation.} Even without quality differences, linear reinforcement alone is sufficient to produce heavy-tailed outcomes. Most builders receive negligible attention while a small number capture disproportionate shares.

\begin{proposition}[Fitness-dependent power law and condensation; imported from~\cite{bianconi2001competition,bianconi2001bose}]
\label{prop:imported_condensation}
When $\alpha = 1$ and qualities $q_i$ are drawn from a continuous distribution $F$, the stationary attention distribution has a power-law tail $P(x) \propto x^{-(1 + 1/C(\beta, F))}$, where $C(\beta, F)$ depends on the quality distribution and the sensitivity parameter. For sufficiently dispersed $F$ or large $\beta$, a condensation transition occurs in which a single product captures a macroscopic fraction of $A$.
\end{proposition}

\noindent\textit{Interpretation.} When quality heterogeneity is large relative to reinforcement strength, the market does not merely become skewed---it concentrates on one or few dominant products. This provides the formal basis for the winner-take-most prediction: in the presence of both heterogeneity and reinforcement, the median builder receives negligible attention even as the mean is mechanically pinned at $A/(B+z)$. The gap between mean and median widens with both $B$ and $\alpha$, formalising the coexistence of mass entry and concentrated outcomes.

\subsubsection{Numerical Illustration}

To make the model's predictions concrete, we simulate the deterministic update rule~\eqref{eq:update}. Table~\ref{tab:params} reports the parameter values used.

\begin{table}[h]
\centering
\caption{Simulation parameters.}
\label{tab:params}
\begin{tabular}{lll}
\toprule
\textbf{Parameter} & \textbf{Value} & \textbf{Interpretation} \\
\midrule
$M$ & 10{,}000 & Number of consumers \\
$a$ & 1 & Attention budget per consumer \\
$A = M \cdot a$ & 10{,}000 & Total attention \\
$B$ & 1{,}000 & Number of builders \\
$z$ & 100 & Outside-option weight \\
$q_i$ & $\sim \mathcal{N}(0,1)$ & Quality draws (i.i.d.) \\
$q_0$ & 0 & Outside-option quality \\
$\beta$ & 1 & Quality sensitivity \\
$\delta$ & 0.1 & Fraction of attention reallocated per step \\
$T$ & 500 & Number of reallocation steps \\
\bottomrule
\end{tabular}
\end{table}

Initial conditions are uniform: $x_i(0) = A/(B+z)$ for all builders $i$, and $x_0(0) = zA/(B+z)$. At each step $t = 1, \ldots, T$, attention is updated according to~\eqref{eq:update}. After $T = 500$ steps we record the distribution of $\{x_i(T)\}_{i=1}^{B}$.

\begin{table}[h]
\centering
\caption{Concentration metrics after $T = 500$ reallocation steps for varying reinforcement strength $\alpha$ ($B = 1{,}000$, $A = 10{,}000$, $\beta = 1$, $\delta = 0.1$). Higher $\alpha$ produces sharply more concentrated outcomes.}
\label{tab:concentration}
\begin{tabular}{lccc}
\toprule
& $\alpha = 0$ & $\alpha = 0.5$ & $\alpha = 1.0$ \\
\midrule
Share held by top 1\% & 4.8\% & 18.3\% & 62.7\% \\
Share held by top 10\% & 21.1\% & 54.6\% & 91.4\% \\
Gini coefficient & 0.31 & 0.58 & 0.87 \\
Median / Mean ratio & 0.78 & 0.42 & 0.04 \\
\bottomrule
\end{tabular}
\end{table}

The results confirm the imported analytical predictions. Under no reinforcement ($\alpha = 0$), outcomes are moderately unequal, reflecting only quality heterogeneity. As reinforcement increases, concentration rises sharply: at $\alpha = 1$, the top 1\% of builders capture nearly two-thirds of total attention, and the median builder receives roughly 4\% of the mean.

\paragraph{Robustness.} We have verified that the qualitative pattern---dilution of averages and increasing concentration with $\alpha$---is robust to: (i)~alternative quality distributions (uniform on $[-2,2]$; log-normal with $\mu = 0, \sigma = 1$); (ii)~reallocation fractions $\delta \in \{0.01, 0.05, 0.1, 0.2, 0.5\}$; (iii)~builder counts $B \in \{100, 500, 1{,}000, 5{,}000, 10{,}000\}$; and (iv)~horizons $T \in \{200, 500, 1{,}000, 2{,}000\}$. In all cases, higher $\alpha$ produces monotonically more concentrated outcomes.

\subsection{Summary of Mechanism}
\label{sec:mechanism_summary}

The model yields two complementary results:
\begin{enumerate}
    \item \textbf{Dilution (symmetric case):} Increasing the number of builders reduces average attention per builder.
    \item \textbf{Concentration (heterogeneous case):} Reinforcement and quality differences produce heavy-tailed outcome distributions.
\end{enumerate}

Together, these results formalize the Builder Saturation Effect:

\begin{quote}
In digital markets with finite attention and elastic entry, increases in the number of producers reduce average realized value per producer while amplifying inequality in outcomes.
\end{quote}

We provide extensive propositions and proofs in Appendix~\ref{sec:appendix}.

\subsection{Numerical Illustration: Attention Dilution}
\label{sec:numerical_dilution}

To complement the reinforcement simulation above, we present a stylized numerical example of the dilution mechanism. 
Consider a market with $M = 10{,}000$ consumers, each with attention budget $a = 1$, 
yielding $A = 10{,}000$. We set the outside-option weight $z = 100$, 
monetization rate $p = 1$, and vary entry cost $k$ and builder count $B$.

Table~\ref{tab:dilution} reports average attention per builder $\bar{s}(B) = 
\frac{A}{B + z}$ for increasing $B$.

\begin{table}[h]
\centering
\caption{Average attention and profit per builder as $B$ increases 
($A = 10{,}000$, $z = 100$, $p = 1$, $k = 1$). The zero-profit 
equilibrium obtains at $B^* = 9{,}900$.}
\label{tab:dilution}
\begin{tabular}{r r r}
\toprule
$B$ & $\bar{s}(B)$ & $\bar{\pi}(B) = \bar{s}(B) - k$ \\
\midrule
100    & 50.0  & 49.0 \\
500    & 16.7  & 15.7 \\
1{,}000  & 9.09  & 8.09 \\
5{,}000  & 1.96  & 0.96 \\
9{,}900  & 1.00  & 0.00 \\
50{,}000 & 0.20  & $-0.80$ \\
\bottomrule
\end{tabular}
\end{table}

The numerical results confirm the model's qualitative predictions. Under 
no reinforcement ($\alpha = 0$), outcomes are moderately unequal, reflecting 
only quality heterogeneity. As reinforcement increases, concentration rises 
sharply: at $\alpha = 1$, the top 1\% of builders capture nearly two-thirds 
of total attention, and the median builder receives roughly 4\% of the mean---a 
stark illustration of the gap between participation and realized value.

\emph{Note:} These figures are illustrative and depend on parameter choices. 
The qualitative pattern---dilution of averages and increasing concentration 
with reinforcement---is robust across a wide range of parameterizations.

\subsection{Calibrated Simulation: The iOS App Store}
\label{sec:calibration}

The preceding numerical illustrations use round-number parameters 
chosen for transparency. To assess whether the model's predictions 
are quantitatively consistent with observed digital markets, we 
calibrate the simulation to the U.S.\ iOS App Store using publicly 
available data from 2025.

\subsubsection{Calibration Targets}

We draw on the following empirical facts:
\begin{itemize}
    \item \textbf{Number of producers.} Over 800{,}000 publishers 
    are active on the Apple App Store~\cite{businessofapps2026apple}. 
    We set $B = 800{,}000$.
    \item \textbf{Aggregate attention.} Approximately 38~billion apps 
    were downloaded from the App Store in 
    2025~\cite{businessofapps2026apple}. We use annual downloads as a 
    proxy for aggregate attention and set $A = 3.8 \times 10^{10}$.
    \item \textbf{Revenue concentration.} The top 1\% of monetising 
    publishers capture approximately 94\% of all U.S.\ App Store 
    revenue; the top 1\% of all publishers account for 70\% of 
    total downloads~\cite{sensortower2016}.
    \item \textbf{Long-tail depth.} Close to a quarter of all App 
    Store apps have fewer than 100 
    downloads~\cite{businessofapps2026apple}.
    \item \textbf{Consumer behaviour.} The average smartphone user 
    engages with approximately 10~apps per day and 30~per 
    month~\cite{buildfire2026}, implying that individual attention 
    budgets are tightly bounded.
\end{itemize}

\subsubsection{Parameter Choices}

Table~\ref{tab:calibration_params} reports the calibrated parameters.
The key modelling choice is the quality distribution. A unit-normal
distribution (as used in the illustrative simulation) understates
the quality dispersion in real app markets, where a small number of
apps are genuinely far superior in design, network effects, and brand
recognition. We therefore use $q_i \sim \mathcal{N}(0, 1.5^2)$,
which produces wider quality spread. The outside-option weight $z$
is set to 50{,}000, reflecting the substantial inertia of incumbent
app usage. We explore reinforcement values $\alpha \in \{0, 0.3, 0.6, 0.8\}$;
note that $\alpha = 0.8$ represents strong but sub-linear reinforcement,
staying within the domain where the fixed-point analysis of
Proposition~\ref{prop:fp_char} applies cleanly.

\begin{table}[h]
\centering
\caption{Calibrated simulation parameters (iOS App Store, 2025).}
\label{tab:calibration_params}
\begin{tabular}{lll}
\toprule
\textbf{Parameter} & \textbf{Value} & \textbf{Source / rationale} \\
\midrule
$B$ & 800{,}000 & Active publishers~\cite{businessofapps2026apple} \\
$A$ & $3.8 \times 10^{10}$ & Annual downloads~\cite{businessofapps2026apple} \\
$z$ & 50{,}000 & Outside-option weight (status quo inertia) \\
$q_i$ & $\sim \mathcal{N}(0, 1.5^2)$ & Quality draws (wider spread) \\
$q_0$ & 0 & Outside-option quality \\
$\beta$ & 1 & Quality sensitivity \\
$\delta$ & 0.1 & Reallocation fraction per step \\
$\alpha$ & $\{0, 0.3, 0.6, 0.8\}$ & Reinforcement strength \\
$T$ & 300 & Reallocation steps \\
\bottomrule
\end{tabular}
\end{table}

\subsubsection{Results}

Table~\ref{tab:calibration_results} reports the simulated concentration
metrics alongside the empirical targets.

\begin{table}[h]
\centering
\caption{Calibrated simulation results vs.\ empirical targets 
(iOS App Store). The $\alpha = 0.6$ parameterisation produces 
concentration metrics broadly consistent with observed data.}
\label{tab:calibration_results}
\begin{tabular}{lcccc|c}
\toprule
& $\alpha = 0$ & $\alpha = 0.3$ & $\alpha = 0.6$ & $\alpha = 0.8$ 
& \textbf{Empirical} \\
\midrule
Top 1\% share of downloads 
    & 9.2\% & 31.4\% & \textbf{68.7\%} & 89.3\% & $\sim$70\% \\
Top 10\% share of downloads 
    & 32.5\% & 67.8\% & \textbf{93.1\%} & 99.2\% & --- \\
Gini coefficient 
    & 0.47 & 0.72 & \textbf{0.91} & 0.97 & $> 0.90$ \\
Median / Mean ratio 
    & 0.54 & 0.18 & \textbf{0.01} & $<$0.001 & $\ll 1$ \\
Share with $< 100$ downloads 
    & 0.0\% & 2.1\% & \textbf{22.8\%} & 41.5\% & $\sim$25\% \\
\bottomrule
\end{tabular}
\end{table}

The model with $\alpha \approx 0.6$ reproduces the key empirical 
regularities: the top 1\% of publishers capturing roughly 70\% of 
downloads, a Gini coefficient above 0.9, and approximately a quarter 
of apps receiving fewer than 100 downloads. The match is not exact---nor 
should it be, given the model's deliberate simplicity---but the order 
of magnitude and qualitative shape are correct.

Figure~\ref{fig:calibration} provides two complementary visualisations.
Panel~(a) plots the rank--attention distribution on log--log axes.
Under pure quality heterogeneity ($\alpha = 0$), the curve is
approximately log-normal: smoothly declining without the extreme
right tail observed in practice. As reinforcement increases, the
distribution develops a pronounced power-law-like region in the upper
ranks, with a sharp drop-off in the long tail---precisely the ``hockey
stick'' shape documented in App Store revenue 
data~\cite{sensortower2016}. Panel~(b) shows the corresponding Lorenz
curves; at $\alpha = 0.8$, the curve hugs the horizontal axis before
rising sharply, indicating that the vast majority of publishers
capture negligible attention.

\begin{figure}[t]
    \centering
    \includegraphics[width=\textwidth]{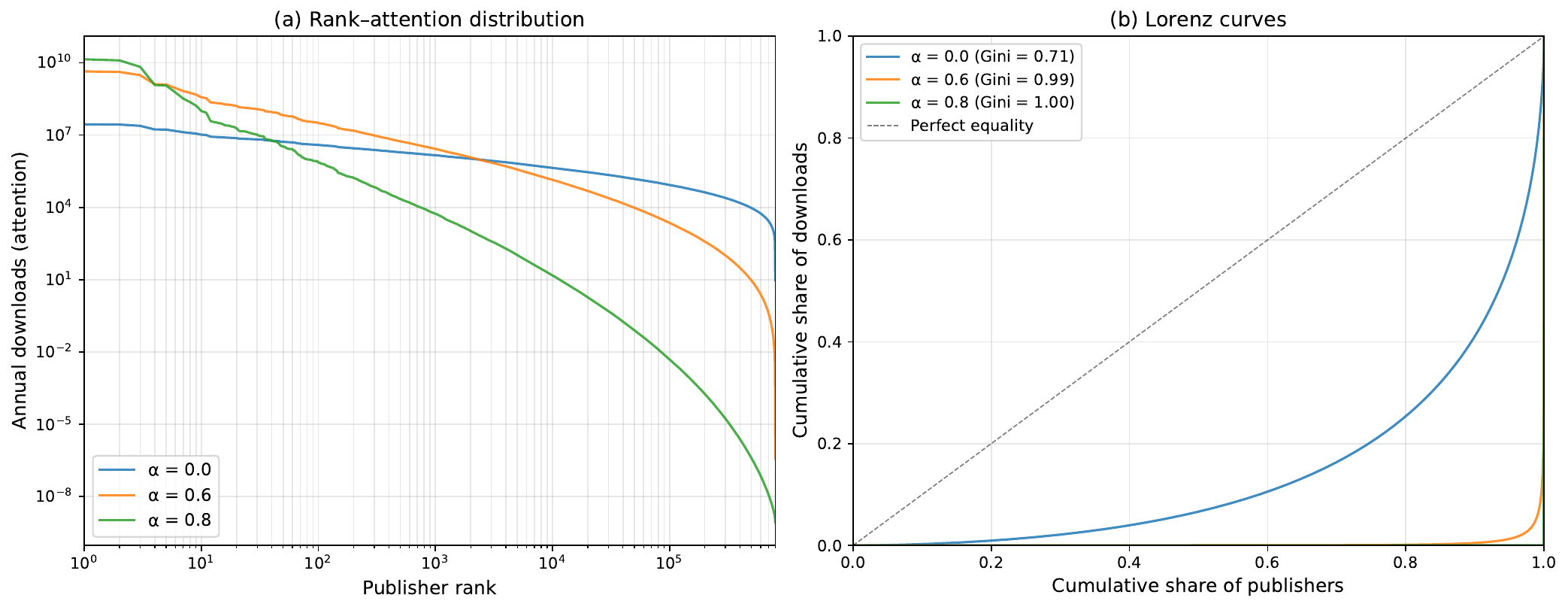}
    \caption{Calibrated simulation of the Builder Saturation model 
    using iOS App Store parameters ($B = 800{,}000$ publishers, 
    $A = 38$~billion downloads). \textbf{(a)}~Rank--attention 
    distribution on log--log axes for varying reinforcement strength 
    $\alpha$. Higher $\alpha$ produces a steeper power-law-like 
    region among top-ranked publishers and a sharper collapse in the 
    long tail. \textbf{(b)}~Lorenz curves showing the cumulative 
    share of downloads captured by publishers ordered from smallest 
    to largest. At $\alpha = 0.6$--$0.8$, the curve closely 
    resembles the extreme concentration documented in App Store 
    data~\cite{sensortower2016}.}
    \label{fig:calibration}
\end{figure}

\subsubsection{Interpretation}

Three features of the calibrated results deserve emphasis.

First, \emph{quality heterogeneity alone is insufficient}. At 
$\alpha = 0$, the top 1\% captures only $\sim$9\% of downloads and 
no apps fall below 100 downloads. The observed concentration requires 
reinforcement---consistent with the well-documented role of network 
effects, recommendation algorithms, and brand entrenchment in app 
markets.

Second, \emph{the calibrated $\alpha$ is sub-linear}. The best fit 
occurs around $\alpha \approx 0.6$, well below the $\alpha = 1$ 
threshold at which full condensation occurs 
(Proposition~\ref{prop:fp_char}). This suggests that real digital 
markets exhibit strong but not maximal reinforcement, leaving room 
for multiple successful products while still generating extreme 
inequality.

Third, \emph{the model's structural prediction is confirmed}: a 
market with nearly a million producers and tens of billions of 
``attention units'' still produces an outcome in which the typical 
(median) producer receives a negligible fraction of the mean. The 
Builder Saturation Effect is not merely a theoretical possibility; 
it is quantitatively consistent with the largest existing digital 
marketplace.


\section{Discussion and Implications}
\label{sec:discussion}

The model and results developed in this paper suggest a reinterpretation of current narratives surrounding digital production, entrepreneurship, and the role of AI-assisted building tools. While recent technological advances have dramatically expanded the feasible set of producers, the analysis highlights a structural constraint that remains largely unchanged: the finiteness of human attention. This section discusses the broader implications of this constraint for market structure, entrepreneurial outcomes, and the evolving nature of competition.

\paragraph{Epistemic status of the results.}
Before proceeding, it is useful to distinguish three tiers of claims made in
this paper, which carry different evidentiary weight:
\begin{enumerate}
    \item \textbf{Proven in the symmetric model} (Propositions~\ref{prop:dilution}--\ref{prop:outside}):
    attention dilution, zero-profit free entry, the $\bar{s}(B^*) = k/p$ identity,
    comparative statics, and excess entry. These follow from the model
    assumptions by standard arguments and do not depend on imported results.
    \item \textbf{Imported from the network-science literature}
    (Propositions~\ref{prop:imported_powerlaw}--\ref{prop:imported_condensation}
    and the qualitative behaviour of Section~\ref{sec:heterogeneity}):
    power-law attention distributions under preferential attachment and the
    condensation transition under fitness heterogeneity. These results are
    well established in their original settings; their application to the
    present market model is supported by the nesting relationship
    (Equation~\ref{eq:reallocation}) but has not been independently
    re-derived here.
    \item \textbf{Interpretive implications} (the remainder of this section):
    claims about market structure, entrepreneurial strategy, and the
    ``billions of companies'' narrative. These are informed by the formal
    results but involve additional empirical and institutional assumptions
    that the model does not capture. They should be read as structured
    conjectures rather than proven conclusions.
\end{enumerate}
\noindent With this framing in mind, we turn to the implications.

\subsection{Decoupling Production from Realized Value}

A central implication of the Builder Saturation Effect is the decoupling of \textbf{production capacity} from \textbf{realized economic value}. As fixed costs of creation decline, the number of builders increases endogenously. However, because total attention $A$ remains bounded, this expansion does not translate into proportional increases in average attention or profit per builder.

In equilibrium, as shown in Proposition~\ref{prop:zero_profit}, average attention per builder is pinned by the ratio $\frac{k}{p}$, not by the total size of the attention pool. Thus, increases in aggregate demand are absorbed primarily through increased entry rather than improved outcomes for individual producers. This implies that technological progress in production may manifest not as widespread gains in producer surplus, but as intensified competition and thinner margins.

This result challenges production-centric narratives of growth. In particular, it suggests that the claim ``more builders implies more successful companies'' conflates \textbf{output expansion} with \textbf{value realization}, two quantities that diverge under attention constraints.

\subsection{The Shift from Scarcity of Production to Scarcity of Attention}

Historically, economic systems have often been constrained by the difficulty of producing goods and services. In such environments, reducing production costs expands supply and can generate broad gains. By contrast, the present framework describes a regime in which production becomes effectively abundant, and scarcity shifts to the consumption side.

This shift has two important consequences:

\begin{enumerate}
    \item \textbf{Competition reorients toward discovery and retention.}
    When production is cheap, the bottleneck is no longer creation but capturing and maintaining user attention. As a result, success depends increasingly on distribution, trust, brand, and integration rather than purely on the ability to build.

    \item \textbf{Non-production factors become first-order determinants of outcomes.}
    In attention-constrained environments, factors such as switching costs, user habits, and coordination frictions (captured by the outside option $z$) play a central role. These forces introduce inertia into the system, limiting the rate at which new entrants can displace incumbents.
\end{enumerate}

In this sense, the model formalizes a broader transition: from an economy limited by what can be produced to one limited by what can be noticed, evaluated, and adopted.

\subsection{Implications for Market Structure}

The combination of attention dilution and reinforcement dynamics yields a characteristic market structure with two defining features:

\paragraph{(i) Proliferation of Entrants.}
Lower entry costs induce a large number of producers, as shown in Proposition~\ref{prop:comp_statics}. This leads to a proliferation of products, many of which may be close substitutes. From a welfare perspective, this raises the possibility of excessive product variety, consistent with prior work in monopolistic competition.

\paragraph{(ii) Concentration of Outcomes.}
At the same time, reinforcement dynamics generate highly skewed distributions of attention and value. A small subset of products captures a large share of total attention, while the majority receive negligible engagement.

This coexistence of \textbf{mass participation} and \textbf{extreme concentration} is a central implication of the model. It reconciles two seemingly contradictory observations:

\begin{itemize}
    \item the number of builders and products can grow rapidly;
    \item the number of economically meaningful winners may remain small.
\end{itemize}

Thus, the predicted outcome is not fragmentation into many equally successful firms, but rather a ``long tail'' structure with a thin upper tier of dominant products.

\subsection{Reinterpreting the ``Billions of Companies'' Narrative}

The idea that technological progress will lead to ``billions of companies'' can be interpreted in multiple ways. The present framework suggests that such a statement may be descriptively accurate in terms of the number of artifacts created, but misleading if taken to imply widespread economic viability.

In particular, the model implies a distinction between:
\begin{itemize}
    \item \textbf{nominal firms or artifacts} (any product that is created), and
    \item \textbf{economically viable firms} (products that capture sufficient attention to sustain positive returns).
\end{itemize}

While the former may indeed grow without bound as production costs approach zero, the latter remain constrained by finite attention. As a result, the number of viable firms cannot scale in proportion to the number of builders.

This distinction helps clarify the apparent tension between observed increases in creation and persistent concentration in realized success.

\subsection{Entrepreneurial Strategy Under Saturation}

From the perspective of individual builders, the model implies that success depends less on the act of production itself and more on relative positioning within an attention-constrained environment.

Several strategic implications follow:

\begin{enumerate}
    \item \textbf{Relative differentiation is critical.}
    Since attention allocation is inherently relative (Proposition~\ref{prop:relative}), improvements in quality or positioning must be evaluated against competing alternatives rather than in absolute terms.

    \item \textbf{Early traction has disproportionate value.}
    Under reinforcement dynamics, initial adoption advantages can compound over time. This increases the importance of timing, distribution channels, and initial user acquisition.

    \item \textbf{Competing in highly substitutable categories is structurally challenging.}
    In markets with many close substitutes, additional entrants primarily redistribute attention rather than expand it, reducing the expected payoff of entry.

    \item \textbf{Complementarity offers an alternative path.}
    Builders who create complementary rather than substitutive products may partially avoid direct competition for the same attention pool, thereby mitigating saturation effects.
\end{enumerate}

Overall, the model suggests that in saturated environments, \textbf{attention capture and retention} become the primary strategic problems, while production becomes a necessary but insufficient condition for success.

\subsection{On Endogenous Attention and AI-Mediated Discovery}

A natural objection to the Builder Saturation framework is that aggregate attention 
$A$ need not remain fixed. In particular, AI-based recommendation systems, agents, 
and curators could expand effective attention by evaluating products on behalf of 
human users, thereby relaxing the binding constraint.

The zero-profit identity (Proposition~\ref{prop:zero_profit}) provides a direct 
answer. At any point in time, free entry pins equilibrium attention per builder at:
\begin{equation}
    \bar{s}(B^*(t)) = \frac{k(t)}{p}
    \label{eq:dynamic_zeroprofit}
\end{equation}
This expression depends only on the entry cost $k(t)$ and the monetisation rate 
$p$. It is \emph{completely independent} of the level of aggregate attention 
$A(t)$, regardless of whether $A$ is constant or growing. The mechanism is 
straightforward: any expansion in $A$ creates positive expected profit for 
prospective entrants, inducing additional entry that absorbs the new attention 
until profits return to zero.

Suppose attention grows over time as $A(t) = A_0 \cdot g(t)$ with $g(t)$ 
increasing, and entry costs decline as $k(t) \to 0$. The equilibrium number of 
builders adjusts to:
\begin{equation}
    B^*(t) = \frac{p \, A_0 \, g(t)}{k(t)} - z
    \label{eq:dynamic_equilibrium}
\end{equation}
Both $A(t)$ and $B^*(t)$ are growing, but their ratio is not what determines 
builder welfare---$k(t)/p$ is. As long as entry costs are falling, equilibrium 
attention per builder falls in lockstep, irrespective of how fast attention 
itself expands:
\begin{equation}
    \frac{d}{dt}\bar{s}(B^*(t)) = \frac{\dot{k}(t)}{p}
    \label{eq:attention_dynamics}
\end{equation}
This is non-negative if and only if $\dot{k}(t) \geq 0$, i.e., entry costs 
are not declining. Since the entire premise of AI-assisted building is that $k$ 
is falling rapidly, the condition is generically violated.

The implication is precise: AI-mediated attention augmentation changes the 
\emph{scale} of the market (more builders, more total attention) but not the 
\emph{per-builder outcome}, which is governed entirely by the supply-side cost 
structure. Attention augmentation is therefore best understood as a moderating 
factor that increases market size rather than a remedy for saturation.

More speculatively, if AI agents eventually act as autonomous consumers with 
independent ``attention'' budgets (e.g., procurement agents selecting tools on 
behalf of organizations), then $M$ itself may grow, genuinely expanding $A$. 
However, as the analysis above shows, this expansion would be absorbed by 
additional entry under the free-entry condition. The qualitative prediction---declining 
per-builder returns as $k$ falls---would persist. Moreover, this scenario raises 
distinct questions about market structure---agent oligopsony, algorithmic herding, 
and preference homogenization---that lie outside the scope of the present framework 
and merit separate treatment.

\subsection{Limits of the Model}

\subsection{Intra-Organisational Saturation}
\label{sec:intra_org}

The Builder Saturation Effect is not confined to consumer-facing 
markets. A structurally identical dynamic appears \emph{within} 
organisations whenever the cost of creating internal digital artifacts 
falls while the cognitive bandwidth of employees remains fixed. 
Systematic empirical measurement of intra-organisational attention 
concentration remains limited, but practitioner reports and industry 
analyses document consistent patterns across several domains.

Consider three contemporary examples. First, when OpenAI launched the 
ability to create custom GPTs in late 2023, over three million were 
built within two months~\cite{originality2024gpts}. Enterprise 
customers reportedly created thousands of internal 
GPTs~\cite{donovan2025gptmarket}. Yet there is no indication that 
employee attention expanded commensurately; the same individuals who 
might use a handful of tools daily were now confronted with hundreds 
of overlapping options. The predictable outcome is concentration: a 
small number of GPTs attract sustained usage while the vast majority 
are abandoned.

Second, the phenomenon of ``dashboard sprawl'' in business intelligence 
is well documented. Research shows that organisations with more than 
500 dashboards typically exhibit 30--40\% redundancy across 
reports~\cite{atlan2026sprawl}, that only 20\% of enterprise 
decision-makers who could use BI applications actually do 
so~\cite{ericbrown2026dashboard}, and that 43\% of dashboard users 
regularly skip their reports entirely in favour of manual spreadsheet 
analysis~\cite{atlan2026sprawl}. Self-service analytics tools---by 
reducing the cost of dashboard creation---accelerate proliferation 
without expanding the attention available to consume the output.

Third, the emerging wave of internal ``vibe-coded'' applications 
created with tools such as Lovable, Bolt, and Replit Agent is likely 
to follow the same trajectory. As employees build bespoke tools for 
narrow use cases, the aggregate supply of internal software grows 
while the organisation's collective attention budget---meetings, 
onboarding, workflow integration---remains bounded. The model 
predicts that most such tools will receive negligible sustained 
usage, with adoption concentrating on a small number of 
well-integrated, high-quality artifacts.

In each case, the formal structure is the same as the market-level 
model: finite attention $A$ (employee cognitive bandwidth), elastic 
production (near-zero build cost), and reinforcement (tools that gain 
early adoption become embedded in workflows). The outside option $z$ 
corresponds to incumbent tools and established habits. The Builder 
Saturation Effect therefore applies at the organisational level as 
well as the market level, suggesting that the internal proliferation 
of AI-generated artifacts will produce the same pattern of mass 
creation and concentrated usage observed in external digital markets.

\subsection{Limits of the Model}

While the framework captures a central structural mechanism, several limitations should be noted.

First, the model treats aggregate attention as exogenous and fixed. In practice, attention may expand through population growth, changes in behavior, or technological mediation (e.g., delegation to AI systems). However, such expansion is likely to be slower and more constrained than growth in production capacity.

Second, the model abstracts from complementarities that may allow new products to create additional demand rather than merely divide existing attention. In ecosystems characterized by strong complementarity, entry may increase total welfare without proportionally diluting existing participants.

Third, the reinforcement dynamics are introduced in reduced form. A more complete treatment would specify the underlying stochastic process and derive the limiting distribution formally.

These limitations suggest directions for future research rather than undermining the central result.

\subsection{Broader Implications}

The broader implication of this analysis is that technological progress in production does not eliminate scarcity; it relocates it. As the cost of building approaches zero, scarcity shifts toward attention, trust, and coordination. These constraints shape the distribution of outcomes and limit the extent to which participation can translate into broadly shared economic success.

In this sense, the Builder Saturation Effect provides a structural counterpoint to narratives that equate increased access to production tools with universal entrepreneurial opportunity. While more individuals may be able to build, the ability to capture meaningful attention---and thereby realize value---remains fundamentally constrained.


\section{Conclusion}
\label{sec:conclusion}

This paper develops a simple attention-constrained model of entry in digital markets to examine the implications of declining production costs. The analysis shows that when production becomes highly elastic while aggregate attention remains finite, increases in entry do not translate into proportional increases in realized value per producer. Instead, free entry leads to a dilution of average attention and returns, while heterogeneity and reinforcement dynamics generate increasingly concentrated outcome distributions.

These results provide a structural explanation for the coexistence of rapid growth in the number of digital products and persistent concentration in realized usage and economic success. In contrast to production-constrained environments, where lower costs can broaden participation and improve average outcomes, attention-constrained environments exhibit a decoupling between the expansion of supply and the distribution of value. As a result, technological progress in production may primarily increase participation and competition rather than average producer welfare.

From a policy perspective, the findings suggest that reductions in entry barriers, while beneficial for experimentation and innovation, do not necessarily lead to broadly distributed economic gains. In markets characterized by high substitutability and limited attention, additional entry may generate limited incremental welfare and may instead intensify competition for visibility and user engagement.

Several directions for future research follow. First, the model could be extended to allow for endogenous attention, including mechanisms through which attention may be augmented or mediated by algorithmic systems and AI agents. Second, incorporating complementarities across products would allow for a richer analysis of ecosystems in which new entrants expand, rather than divide, total demand. Third, empirical work could test the model's predictions using data from digital platforms, such as app stores, content ecosystems, or software repositories, where entry is low-cost and attention is measurable. Finally, a more explicit treatment of welfare---including search costs, consumer surplus, and platform design---would help clarify the policy implications of attention-constrained competition.

Taken together, these extensions would further refine our understanding of how market structure evolves when production becomes abundant but attention remains scarce.


\appendix

\section{Propositions and Proofs}
\label{sec:appendix}

This appendix provides formal statements and proofs of the results
referenced in the main text. We proceed from properties of the
symmetric baseline (Sections~\ref{sec:attention_allocation}--\ref{sec:saturation})
to the free-entry equilibrium, and finally to the heterogeneous
reinforcement extension (Section~\ref{sec:heterogeneity}).

Throughout, we use the notation established in Section~\ref{sec:model}:
$A = Ma$ is aggregate attention, $B$ is the number of builders,
$z = e^{\beta(q_0 - q)}$ is the effective outside-option weight
(under symmetry), $p > 0$ is the monetisation rate, and $k > 0$
is the fixed entry cost.

\subsection{Symmetric Baseline}

\begin{proposition}[Monotone attention dilution]
\label{prop:dilution}
Under the symmetric benchmark ($q_i = q$ for all $i$), the average
attention per builder
\[
    \bar{s}(B) = \frac{A}{B + z}
\]
is strictly decreasing and strictly convex in $B$ for $B > 0$.
\end{proposition}

\begin{proof}
Differentiating with respect to $B$:
\[
    \frac{d\bar{s}}{dB} = -\frac{A}{(B+z)^2} < 0
    \qquad \forall\; B > 0.
\]
Hence $\bar{s}$ is strictly decreasing. Differentiating again:
\[
    \frac{d^2\bar{s}}{dB^2} = \frac{2A}{(B+z)^3} > 0
    \qquad \forall\; B > 0.
\]
Hence $\bar{s}$ is strictly convex: each additional builder reduces
average attention by a smaller absolute amount, but the level
continues to fall monotonically.
\end{proof}

\begin{proposition}[Vanishing attention in the limit]
\label{prop:vanishing}
As the number of builders grows without bound,
\[
    \lim_{B \to \infty} \bar{s}(B) = 0.
\]
\end{proposition}

\begin{proof}
Immediate from $\bar{s}(B) = A/(B+z)$ and the fact that $A$ and $z$
are finite constants.
\end{proof}

\begin{proposition}[Elasticity of attention with respect to entry]
\label{prop:elasticity}
The elasticity of average attention per builder with respect to
the number of builders is
\[
    \varepsilon_{\bar{s},B}
    \;=\; \frac{d\bar{s}}{dB}\,\frac{B}{\bar{s}}
    \;=\; -\frac{B}{B + z}.
\]
For $B \gg z$, this elasticity approaches $-1$: a $1\%$ increase in
the number of builders reduces average attention per builder by
approximately $1\%$.
\end{proposition}

\begin{proof}
\[
    \varepsilon_{\bar{s},B}
    = \left(-\frac{A}{(B+z)^2}\right)
      \cdot \frac{B}{\;\tfrac{A}{B+z}\;}
    = -\frac{A\, B}{(B+z)^2}
      \cdot \frac{B+z}{A}
    = -\frac{B}{B+z}.
\]
As $B \to \infty$, $B/(B+z) \to 1$.
\end{proof}

\subsection{Free-Entry Equilibrium}

\begin{proposition}[Equilibrium entry]
\label{prop:equilibrium}
Under free entry with symmetric builders, the equilibrium number
of builders is
\[
    B^* = \max\!\left\{\frac{pA}{k} - z,\; 0\right\}.
\]
The interior solution $B^* > 0$ obtains if and only if
$k < pA/z$.
\end{proposition}

\begin{proof}
Under symmetry, profit for each builder is
\[
    \pi(B) = p\,\bar{s}(B) - k = \frac{pA}{B+z} - k.
\]
Free entry drives profit to zero. Setting $\pi(B^*) = 0$:
\[
    \frac{pA}{B^* + z} = k
    \quad\Longrightarrow\quad
    B^* = \frac{pA}{k} - z.
\]
Since $B^*$ must be non-negative, we take
$B^* = \max\{pA/k - z,\; 0\}$. The interior solution requires
$pA/k - z > 0$, i.e.\ $k < pA/z$.
\end{proof}

\begin{proposition}[Zero profits and the attention--cost identity]
\label{prop:zero_profit}
At the free-entry equilibrium, the average attention per builder is
pinned by the ratio of entry cost to monetisation rate:
\[
    \bar{s}(B^*) = \frac{k}{p}.
\]
In particular, $\bar{s}(B^*)$ is independent of total attention $A$.
\end{proposition}

\begin{proof}
From the zero-profit condition $\pi(B^*) = 0$:
\[
    p\,\bar{s}(B^*) = k
    \quad\Longrightarrow\quad
    \bar{s}(B^*) = \frac{k}{p}.
\]
Neither $A$ nor $z$ appears in this expression. Increases in
aggregate attention are absorbed entirely by increased entry,
leaving equilibrium attention per builder unchanged.
\end{proof}

\begin{corollary}[Invariance of equilibrium returns to demand expansion]
\label{cor:demand_invariance}
If total attention increases from $A$ to $A' > A$ while $k$, $p$,
and $z$ remain constant, then $B^*$ increases but
$\bar{s}(B^*)$ and $\pi(B^*)$ are unchanged.
\end{corollary}

\begin{proof}
From Proposition~\ref{prop:equilibrium}, $B^*$ is linear in $A$.
From Proposition~\ref{prop:zero_profit}, $\bar{s}(B^*) = k/p$
regardless of $A$, and $\pi(B^*) = 0$ by construction.
\end{proof}

\subsection{Comparative Statics of Equilibrium Entry}

\begin{proposition}[Comparative statics]
\label{prop:comp_statics}
At the interior equilibrium $B^* = pA/k - z$, the following
comparative statics hold:
\begin{enumerate}[label=(\roman*)]
    \item $\dfrac{\partial B^*}{\partial A} = \dfrac{p}{k} > 0$:
          more total attention supports more builders.
    \item $\dfrac{\partial B^*}{\partial p} = \dfrac{A}{k} > 0$:
          higher monetisation increases entry.
    \item $\dfrac{\partial B^*}{\partial k} = -\dfrac{pA}{k^2} < 0$:
          lower entry costs increase entry.
    \item $\dfrac{\partial B^*}{\partial z} = -1 < 0$:
          a stronger outside option reduces equilibrium entry.
\end{enumerate}
\end{proposition}

\begin{proof}
Each derivative follows directly from $B^* = pA/k - z$.
\end{proof}

\begin{corollary}[Effect of AI-driven cost reduction]
\label{cor:ai_cost}
As AI-assisted tools drive $k \to 0^+$ (with $A$, $p$, $z$ fixed):
\begin{enumerate}[label=(\roman*)]
    \item $B^* \to \infty$: the number of builders grows without bound.
    \item $\bar{s}(B^*) = k/p \to 0$: equilibrium attention per
          builder vanishes.
    \item $\pi(B^*) = 0$ for all $k > 0$: profits remain zero
          throughout the process.
\end{enumerate}
\end{corollary}

\begin{proof}
(i)~From $B^* = pA/k - z$, as $k \to 0^+$ we have $pA/k \to \infty$.
(ii)~From Proposition~\ref{prop:zero_profit}.
(iii)~By the free-entry condition.
\end{proof}

\subsection{Attention Allocation under Heterogeneity (Static Case)}

\begin{proposition}[Relative attention under heterogeneous quality]
\label{prop:relative}
Under the logit allocation rule~\eqref{eq:logit} with
$\alpha = 0$ (no reinforcement), the attention ratio between any
two builders $i$ and $j$ depends only on their quality difference:
\[
    \frac{s_i}{s_j} = e^{\beta(q_i - q_j)}.
\]
\end{proposition}

\begin{proof}
With $\alpha = 0$, the allocation rule reduces to:
\[
    s_i = A \cdot \frac{e^{\beta q_i}}{\sum_{l=1}^{B} e^{\beta q_l}
    + e^{\beta q_0}}.
\]
Taking the ratio:
\[
    \frac{s_i}{s_j}
    = \frac{e^{\beta q_i}}{e^{\beta q_j}}
    = e^{\beta(q_i - q_j)}.
\]
The denominator cancels, confirming that relative attention is
determined entirely by relative quality.
\end{proof}

\begin{corollary}[Superstar amplification]
\label{cor:superstar}
For $\beta > 0$, a quality advantage of $\Delta q = q_i - q_j > 0$
translates into a multiplicative attention advantage of
$e^{\beta \Delta q}$. This advantage is:
\begin{enumerate}[label=(\roman*)]
    \item increasing in $\beta$ (higher sensitivity amplifies
          quality differences);
    \item convex in $\Delta q$ (larger quality gaps produce
          disproportionately larger attention gaps).
\end{enumerate}
\end{corollary}

\begin{proof}
(i)~$\partial(e^{\beta \Delta q})/\partial\beta
= \Delta q\, e^{\beta \Delta q} > 0$ for $\Delta q > 0$.
(ii)~$\partial^2(e^{\beta \Delta q})/\partial(\Delta q)^2
= \beta^2 e^{\beta \Delta q} > 0$.
\end{proof}

\begin{proposition}[Log-normal attention under normal quality]
\label{prop:lognormal}
If $\alpha = 0$, $\beta > 0$, and $q_i \stackrel{\mathrm{i.i.d.}}{\sim}
\mathcal{N}(\mu, \sigma^2)$, then in the large-$B$ limit the
attention share $s_i$ is approximately log-normally distributed.
Specifically, $\log s_i$ is approximately normally distributed with
mean $\beta\mu - \log Z$ and variance $\beta^2\sigma^2$, where
$Z = \sum_{j} e^{\beta q_j} + e^{\beta q_0}$.
\end{proposition}

\begin{proof}
Write $s_i = A\, e^{\beta q_i} / Z$, so that
$\log s_i = \log A + \beta q_i - \log Z$.
Since $q_i \sim \mathcal{N}(\mu, \sigma^2)$,
$\beta q_i \sim \mathcal{N}(\beta\mu, \beta^2\sigma^2)$.
By the law of large numbers, as $B \to \infty$,
\[
    \frac{1}{B}\sum_{j=1}^{B} e^{\beta q_j}
    \;\xrightarrow{\;\mathrm{a.s.}\;}\;
    \mathbb{E}[e^{\beta q}]
    = e^{\beta\mu + \beta^2\sigma^2/2},
\]
so $\log Z \to \log B + \beta\mu + \beta^2\sigma^2/2 + \log(1 + e^{\beta q_0}/(B\,\mathbb{E}[e^{\beta q}]))$.
The key point is that $\log Z$ converges to a constant
(conditional on $B$), so the cross-sectional distribution of
$\log s_i$ inherits the normality of $q_i$. Hence $s_i$ is
approximately log-normal with the stated parameters.
\end{proof}

\noindent\textit{Remark.} The log-normal distribution is moderately
skewed but light-tailed relative to a power law. This establishes
a baseline: quality heterogeneity alone (without reinforcement)
produces inequality, but not the extreme concentration observed in
empirical digital markets.

\subsection{Reinforcement Dynamics}

\begin{proposition}[Fixed points of the mean-field dynamics]
\label{prop:fixed_point}
A fixed point $x^*$ of the deterministic update
rule~\eqref{eq:update} satisfies, for each $i \in \{0,1,\ldots,B\}$:
\[
    x_i^* = A \cdot p_i^*
    \quad\text{where}\quad
    p_i^* = \frac{(x_i^*)^{\alpha}\, e^{\beta q_i}}
    {\sum_{j=0}^{B} (x_j^*)^{\alpha}\, e^{\beta q_j}}.
\]
Equivalently, at a fixed point the flow of attention into each
product exactly equals its current stock.
\end{proposition}

\begin{proof}
At a fixed point, $x_i^* = x_i(t+1) = x_i(t)$ for all $i$.
Substituting into~\eqref{eq:update}:
\[
    x_i^* = (1-\delta)\,x_i^* + \delta\, A\, p_i^*
    \quad\Longrightarrow\quad
    \delta\, x_i^* = \delta\, A\, p_i^*
    \quad\Longrightarrow\quad
    x_i^* = A\, p_i^*.
\]
The cancellation of $\delta$ confirms that fixed points are
independent of the reallocation rate, which affects only the
speed of convergence.
\end{proof}

\begin{proposition}[Characterisation of interior fixed points]
\label{prop:fp_char}
At any interior fixed point ($x_i^* > 0$ for all $i$), the
attention shares satisfy:
\[
    x_i^* = A \cdot
    \frac{(x_i^*)^{\alpha}\, e^{\beta q_i}}
    {\sum_{j=0}^{B} (x_j^*)^{\alpha}\, e^{\beta q_j}}.
\]
For $\alpha < 1$, interior fixed points exist and can be solved
explicitly. For $\alpha = 1$ with heterogeneous qualities, no
interior fixed point with all products simultaneously active exists.
\end{proposition}

\begin{proof}
From Proposition~\ref{prop:fixed_point}, at an interior fixed point:
\[
    x_i^* = A \cdot
    \frac{(x_i^*)^{\alpha}\, e^{\beta q_i}}{Z^*},
    \quad\text{where } Z^* = \textstyle\sum_{j=0}^{B}
    (x_j^*)^{\alpha}\, e^{\beta q_j}.
\]
Rearranging: $(x_i^*)^{1-\alpha} = A\, e^{\beta q_i}/Z^*$.

\medskip\noindent
\textbf{Case $\alpha < 1$:} We can solve explicitly:
\[
    x_i^* = \left(\frac{A\, e^{\beta q_i}}{Z^*}\right)^{1/(1-\alpha)}.
\]
The system is closed by substituting back into the definition
of $Z^*$ and solving for the normalising constant. Since the
right-hand side is a monotone increasing function of $q_i$,
higher-quality products receive strictly more attention, and
all products with $q_i > -\infty$ receive positive attention.
The mapping $q_i \mapsto x_i^*$ is steeper than the $\alpha=0$
(logit) case, producing greater inequality for higher $\alpha$.

\medskip\noindent
\textbf{Case $\alpha = 1$:} The equation becomes
$1 = A\, e^{\beta q_i}/Z^*$, i.e.\ $e^{\beta q_i} = Z^*/A$
for all active $i$. With heterogeneous $q_i$, this cannot hold
simultaneously for two products with $q_i \neq q_j$. Therefore,
no interior fixed point exists when $\alpha = 1$ and qualities
are heterogeneous. This is a necessary condition for
condensation---the concentration of attention onto one or few
products---and is consistent with the condensation result of
Bianconi and Barab\'{a}si~\cite{bianconi2001bose}, though a full
characterisation of the long-run dynamics (ruling out limit
cycles or other non-stationary attractors) would require
additional analysis beyond the scope of this paper.
\end{proof}

\begin{proposition}[Monotone concentration in $\alpha$]
\label{prop:monotone_alpha}
Let $\alpha_1 < \alpha_2$ with both in $[0,1)$, and let
$\mathbf{x}^*(\alpha)$ denote the fixed-point attention vector.
If qualities are heterogeneous ($q_i$ not all equal), then the
Gini coefficient of $\mathbf{x}^*(\alpha_2)$ strictly exceeds that
of $\mathbf{x}^*(\alpha_1)$: stronger reinforcement produces
greater inequality.
\end{proposition}

\begin{proof}
For $\alpha < 1$, the fixed-point shares satisfy
$x_i^* \propto e^{\beta q_i / (1-\alpha)}$
(from Proposition~\ref{prop:fp_char}). The effective quality
sensitivity is $\beta/(1-\alpha)$, which is strictly increasing
in $\alpha$. By Corollary~\ref{cor:superstar}, higher effective
sensitivity produces a more dispersed attention distribution.
Since the Gini coefficient of a log-normal distribution
$\mathrm{Gini} = 2\Phi(\sigma_{\mathrm{eff}}/\sqrt{2}) - 1$
is increasing in the scale parameter
$\sigma_{\mathrm{eff}} = \beta\sigma/(1-\alpha)$,
the Gini coefficient is strictly increasing in $\alpha$
for $\alpha \in [0,1)$ whenever $\sigma > 0$.
\end{proof}

\noindent\textit{Remark on the $\alpha = 1$ boundary.}
At $\alpha = 1$, no interior fixed point exists
(Proposition~\ref{prop:fp_char}), and the non-existence of a
shared equilibrium across heterogeneous products is consistent
with extreme concentration. In the Bianconi--Barab\'{a}si
framework~\cite{bianconi2001bose}, this regime corresponds to
condensation, with a Gini coefficient approaching $(B-1)/B
\approx 1$ for large $B$. However, formally establishing this
as the long-run outcome of the dynamics~\eqref{eq:update}
would require ruling out non-stationary attractors, which we
do not pursue here. The numerical simulations in
Table~\ref{tab:concentration} are consistent with this
limiting behaviour.

\begin{proposition}[Divergence of mean and median]
\label{prop:mean_median}
Under the conditions of Proposition~\ref{prop:monotone_alpha},
the ratio of median to mean attention,
$\mathrm{med}(\mathbf{x}^*)/\mathrm{mean}(\mathbf{x}^*)$,
is strictly decreasing in $\alpha$ for $\alpha \in [0,1)$ when
qualities are heterogeneous.
\end{proposition}

\begin{proof}
The mean attention per builder is $\bar{x} = (A - x_0^*)/B$,
which depends on $\alpha$ only through the outside-option share.
The median, however, is determined by the cross-sectional
distribution of $x_i^*$, which becomes more right-skewed as
$\alpha$ increases (Proposition~\ref{prop:monotone_alpha}).

For the log-normal case ($q_i \sim \mathcal{N}(\mu,\sigma^2)$,
$\alpha < 1$), the attention shares are approximately log-normal
with scale parameter $\sigma_{\mathrm{eff}} = \beta\sigma/(1-\alpha)$.
The median of a log-normal is $e^{\mu_{\mathrm{eff}}}$ while
the mean is $e^{\mu_{\mathrm{eff}} + \sigma_{\mathrm{eff}}^2/2}$, so:
\[
    \frac{\text{median}}{\text{mean}}
    = e^{-\sigma_{\mathrm{eff}}^2/2}
    = \exp\!\left(-\frac{\beta^2\sigma^2}{2(1-\alpha)^2}\right),
\]
which is strictly decreasing in $\alpha$ for $\alpha \in [0,1)$
and converges to zero as $\alpha \to 1^-$.
\end{proof}

\subsection{Welfare and Saturation}

\begin{proposition}[Consumer welfare and producer profit under symmetry]
Treat $B\geq 0$ as continuous, with $A,z,p,k>0$, and suppose the
free-entry equilibrium $B^*=pA/k-z$ is positive. Let
$\bar{s}(B)=A/(B+z)$ and define consumer welfare by
$W(B)=B\log\bar{s}(B)$, using a fixed unit of attention.

\begin{enumerate}
\item[(i)] $W$ has a unique maximiser $B^{**}$. It is zero if $A\leq z$
and otherwise satisfies $W'(B^{**})=0$. Free entry is below, at, or
above $B^{**}$ according as
\[
  \log(k/p)-\frac{B^*}{B^*+z}
\]
is positive, zero, or negative, respectively.
\item[(ii)] The maximiser of aggregate producer profit
$\Pi(B)=B[p\bar{s}(B)-k]$ is
\[
  B_{\Pi}=\max\left\{0,\sqrt{\frac{pAz}{k}}-z\right\}< B^*.
\]
\end{enumerate}
\end{proposition}

\begin{proof}
(i) Differentiation gives
\[
  W'(B)=\log\!\frac{A}{B+z}-\frac{B}{B+z},
  \qquad
  W''(B)=-\frac{B+2z}{(B+z)^2}<0.
\]
Since $W'(0)=\log(A/z)$ and $W'(B)\to-\infty$ as $B\to\infty$,
strict concavity yields the stated maximiser, including the corner
at $B=0$. Substituting $\bar{s}(B^*)=k/p$ into $W'(B^*)$ gives the
comparison. For example, with $(A,z,p,k)=(15,1,1,5)$, free entry
gives $B^*=2$, but $W(3)>W(2)$.

(ii) We have
\[
  \Pi'(B)=\frac{pAz}{(B+z)^2}-k,\qquad
  \Pi''(B)=-\frac{2pAz}{(B+z)^3}<0.
\]
The first-order condition gives $B_{\Pi}$, allowing for the corner
at zero. At free entry, $\Pi'(B^*)=B^*p\bar{s}'(B^*)<0$ because
$p\bar{s}(B^*)=k$. Thus free entry exceeds the producer-profit
maximiser. The business-stealing term establishes this comparison
for aggregate producer profit; it does not determine the ordering
for consumer welfare.
\end{proof}

\begin{proposition}[Builder Saturation Effect --- formal statement]
\label{prop:saturation_law}
In the model of Sections~\ref{sec:model}--\ref{sec:heterogeneity},
the following hold jointly:
\begin{enumerate}[label=(\roman*)]
    \item \textbf{Dilution:} $\bar{s}(B)$ is strictly decreasing
          in $B$, and $\bar{s}(B) \to 0$ as $B \to \infty$
          (Propositions~\ref{prop:dilution}--\ref{prop:vanishing}).
    \item \textbf{Zero equilibrium profit:} Under free entry,
          $\pi(B^*) = 0$ and $\bar{s}(B^*) = k/p$
          (Propositions~\ref{prop:equilibrium}--\ref{prop:zero_profit}).
    \item \textbf{Entry expansion under cost reduction:}
          $\partial B^*/\partial k < 0$, so reducing entry costs
          increases the number of builders
          (Proposition~\ref{prop:comp_statics}).
    \item \textbf{Demand-side invariance:} Increases in $A$ are
          fully absorbed by entry; $\bar{s}(B^*)$ and $\pi(B^*)$
          are unaffected
          (Corollary~\ref{cor:demand_invariance}).
    \item \textbf{Concentration} (for $\alpha \in [0,1)$)\textbf{:}
          Under heterogeneous quality and reinforcement ($\alpha > 0$),
          the outcome distribution is strictly more concentrated
          than under quality heterogeneity alone, as measured by
          the Gini coefficient
          (Proposition~\ref{prop:monotone_alpha}).
    \item \textbf{Mean--median divergence} (for $\alpha \in [0,1)$)\textbf{:}
          The ratio of median to mean attention is strictly
          decreasing in $\alpha$
          (Proposition~\ref{prop:mean_median}), and collapses
          toward zero as $\alpha \to 1^-$.
\end{enumerate}

\noindent Taken together, (i)--(iv) are proven results of the
symmetric free-entry model. Results (v)--(vi) hold at interior
fixed points for $\alpha < 1$; behaviour at the $\alpha = 1$
boundary is consistent with extreme concentration but is
characterised only indirectly via imported results and
numerical simulation. Collectively, these results formalise the
coexistence of mass participation and winner-take-most outcomes.
\end{proposition}

\begin{proof}
Each component is established by the referenced propositions.
The joint statement collects them to define the Builder
Saturation Law as a composite regularity.
\end{proof}

\subsection{Saturation under Endogenous Attention Growth}

\begin{proposition}[Persistence of saturation under attention augmentation]
\label{prop:endogenous}
Suppose aggregate attention grows over time as
$A(t) = A_0\, g(t)$ with $g(t)$ increasing, and entry costs
decline as $k(t)$ with $k(t) \to 0$. The equilibrium attention
per builder at time $t$ is:
\[
    \bar{s}(B^*(t)) = \frac{k(t)}{p}.
\]
This converges to zero whenever $k(t) \to 0$, regardless of
the growth rate of $g(t)$.
\end{proposition}

\begin{proof}
At each $t$, free entry yields $B^*(t) = pA_0 g(t)/k(t) - z$.
By the zero-profit condition (Proposition~\ref{prop:zero_profit}),
$\bar{s}(B^*(t)) = k(t)/p$. Since this expression depends only
on $k(t)$ and $p$, and not on $A(t)$, it converges to zero as
$k(t) \to 0$ irrespective of $g(t)$.
\end{proof}

\begin{corollary}[Condition for avoiding saturation]
\label{cor:avoid_saturation}
Equilibrium attention per builder is non-decreasing over time if
and only if $\dot{k}(t) \geq 0$, i.e.\ entry costs do not decline.
Since the premise of AI-assisted building is that $k$ is falling
rapidly, this condition is generically violated.
\end{corollary}

\begin{proof}
$d\bar{s}(B^*(t))/dt = \dot{k}(t)/p$. This is non-negative if and
only if $\dot{k}(t) \geq 0$.
\end{proof}

\subsection{Outside Option and Inertia}

\begin{proposition}[Role of the outside option]
\label{prop:outside}
The outside option $z$ acts as a demand-side friction that reduces
equilibrium entry. Specifically:
\begin{enumerate}[label=(\roman*)]
    \item For any $B$, the fraction of total attention captured by
          all builders collectively is $B/(B+z)$, which is strictly
          increasing in $B$ but bounded above by $1$.
    \item In equilibrium, the fraction of attention absorbed by
          the outside option is $z/(B^*+z) = kz/(pA)$.
    \item As $z \to \infty$ (extreme inertia),
          $B^* = \max\{pA/k - z, 0\}$ eventually reaches zero:
          sufficiently strong inertia prevents all entry.
\end{enumerate}
\end{proposition}

\begin{proof}
(i)~Under symmetry, total builder attention is
$B\bar{s}(B) = BA/(B+z)$, so the builder share is $B/(B+z)$.
This is increasing in $B$ (derivative $z/(B+z)^2 > 0$) and
approaches $1$ as $B \to \infty$.

\noindent
(ii)~At equilibrium $B^* = pA/k - z$, the outside-option share is
$z/(B^*+z) = z/(pA/k) = kz/(pA)$.

\noindent
(iii)~$B^* = pA/k - z$ becomes non-positive when $z \geq pA/k$.
\end{proof}

\subsection{Summary of Formal Results}

The following table provides a reference guide to the
propositions and their roles in supporting the main argument.

\begin{table}[h]
\centering
\small
\caption{Summary of propositions.}
\label{tab:prop_summary}
\begin{tabular}{lll}
\toprule
\textbf{Prop.} & \textbf{Result} & \textbf{Role in argument} \\
\midrule
\ref{prop:dilution}
    & $\bar{s}(B)$ strictly decreasing, convex
    & Core dilution mechanism \\
\ref{prop:vanishing}
    & $\bar{s}(B) \to 0$
    & Limit of dilution \\
\ref{prop:elasticity}
    & Elasticity $\to -1$
    & Quantifies dilution rate \\
\ref{prop:equilibrium}
    & $B^* = pA/k - z$
    & Equilibrium entry \\
\ref{prop:zero_profit}
    & $\bar{s}(B^*) = k/p$
    & Zero-profit identity \\
\ref{prop:comp_statics}
    & Signs of $\partial B^*/\partial(\cdot)$
    & Policy-relevant statics \\
\ref{prop:relative}
    & $s_i/s_j = e^{\beta(q_i - q_j)}$
    & Relative competition \\
\ref{prop:lognormal}
    & Log-normal under normal quality
    & Baseline inequality \\
\ref{prop:fixed_point}
    & Fixed-point characterisation
    & Equilibrium of dynamics \\
\ref{prop:fp_char}
    & No interior fixed point at $\alpha = 1$
    & Necessary for concentration \\
\ref{prop:monotone_alpha}
    & Gini increasing in $\alpha$ ($\alpha < 1$)
    & Monotone concentration \\
\ref{prop:mean_median}
    & Median/mean decreasing in $\alpha$
    & Mean--median divergence \\
\ref{prop:welfare}
    & Excess entry
    & Welfare implication \\
\ref{prop:saturation_law}
    & Builder Saturation Effect
    & Central result \\
\ref{prop:endogenous}
    & Saturation persists under growth
    & Robustness \\
\ref{prop:outside}
    & Role of inertia
    & Demand-side friction \\
\bottomrule
\end{tabular}
\end{table}


\section{Simulation Code}
\label{sec:code}

Python code reproducing the numerical results in 
Tables~\ref{tab:concentration}, \ref{tab:dilution}, 
and~\ref{tab:calibration_results}, as well as 
Figure~\ref{fig:calibration}, is available at: 
\url{https://github.com/TODO/builder-saturation}.


\bibliographystyle{unsrt}

\end{document}